\documentclass[a4paper,abstracton]{scrartcl}

\usepackage[utf8]{inputenc}

\usepackage{amsmath}
\usepackage{amsthm}
\usepackage{amssymb}

\usepackage[usenames,dvipsnames]{color}
\usepackage{graphicx}
\usepackage{caption, subcaption}

\usepackage{tikz}
\usetikzlibrary{automata}
\usetikzlibrary{arrows,shapes}

\usepackage{bbding}

\usepackage{amssymb,amsmath,verbatim,tabularx,enumitem,colonequals,graphicx,psfrag}

\usepackage{cite}

\newcommand{\cU}{\mathcal{U}}

\newcommand{\R}{\mathbb{R}}
\newcommand{\X}{\mathcal{X}}

\usepackage{authblk}
\usepackage[normalem]{ulem} 

\usepackage{framed}
\usepackage{rotating}
\usepackage{multirow}

\newtheorem{theorem}{Theorem}
\newtheorem{lemma}[theorem]{Lemma}

\newtheorem{corollary}[theorem]{Corollary}

\newtheorem{example}[theorem]{Example}

\newcommand{\BIGOP}[1]{\mathop{\mathchoice%
{\raise-0.22em\hbox{\huge $#1$}}%
{\raise-0.05em\hbox{\Large $#1$}}{\hbox{\large $#1$}}{#1}}}

\allowdisplaybreaks

\begin{document}

\title{On Scenario Aggregation to Approximate Robust Optimization Problems}

\author[1]{Andr\'e Chassein}
\affil[1]{Fachbereich Mathematik, University of Kaiserslautern, Germany}

\author[2]{Marc Goerigk}
\affil[2]{Department of Management Science, Lancaster University, United Kingdom}

\date{}

\maketitle

\begin{abstract}
As most robust combinatorial min-max and min-max regret problems with discrete uncertainty sets are NP-hard, research into approximation algorithm and approximability bounds has been a fruitful area of recent work. A simple and well-known approximation algorithm is the midpoint method, where one takes the average over all scenarios, and solves a problem of nominal type. Despite its simplicity, this method still gives the best-known bound on a wide range of problems, such as robust shortest path, or robust assignment problems.

In this paper we present a simple extension of the midpoint method based on scenario aggregation, which improves the current best $K$-approximation result to an $(\varepsilon K)$-approximation for any desired $\varepsilon > 0$. Our method can be applied to min-max as well as min-max regret problems.
\end{abstract}


\section{Introduction}

We consider uncertain optimization problems of the form
\[ \min \{c^t x : x\in\X\subseteq\R^n \} \tag{P} \]
where $\X$ is the set of feasible solutions, and $c$ is an uncertain objective function that comes from some uncertainty set $\cU$. Two popular approaches to reformulate such an uncertain problem to a robust counterpart are min-max optimization
\[ \min_{x\in\X} \max_{c\in\cU} c^t x \]
and min-max regret optimization
\[ \min_{x\in\X} \max_{c\in\cU} \left( c^t x - opt(c) \right)\]
where $opt(c) = \min_{y\in\X} c^ty$ is used as an additional normalization term. In particular for combinatorial problems where $\X\subseteq\{0,1\}^n$, problems of this type have received significant attention in the research literature, see, e.g., the surveys \cite{KouYu97,Aissi2009,kasperski2016robust} on this topic. In this paper we focus on the case of discrete uncertainty, i.e., the uncertainty set is of the form $\cU = \{ c_1,\ldots,c_K\} \subset \mathbb{R}^n$.

For most combinatorial problems where the deterministic version can be solved in polynomial time (e.g., shortest path, spanning tree, selection), both robust counterparts turn out to be NP-hard. Therefore, the approximability of such problems has been analyzed (see, e.g., \cite{Aissi2007281}). 

A popular approximation algorithm due to its generality and simplicity is the midpoint method (see, e.g., 
\cite{chassein2015new}). The idea is to define a new scenario $\hat{c} := \frac{1}{K} \sum_{i\in[K]} c_i$, which is the average of all scenarios in the uncertainty set, and to solve a nominal problem with respect to these costs. This method is known to be a $K$-approximation algorithm for both min-max and min-max regret optimization. In the case of interval uncertainty, this approach even gives a 2-approximation \cite{kasperski2006approximation,conde2012constant}. Quite surprisingly, this is still the best known approximation guarantee for several problems, see Table~\ref{results}. In column ''Fix $K$'', we denote if an FPTAS is known for the problem with fixed number of scenarios $K$.

\begin{table}
\centering
\begin{tabular}{r|r|cc|c}
 & Problem & \phantom{additional}LB\phantom{additional} & \phantom{additional}UB\phantom{additional} & Fix $K$\\
\hline
\parbox[t]{2mm}{\multirow{5}{*}{\rotatebox[origin=c]{90}{Min-Max}}} & Shortest Path & $\mathcal{O}(\log^{1-\varepsilon} K)$ & $K$  &  \Checkmark \\
 & Spanning Tree & $\mathcal{O}(\log^{1-\varepsilon} K)$ & $\mathcal{O}(\log^2 n)$ & \Checkmark\\
 & $s$-$t$ Cut & $\mathcal{O}(\log^{1-\varepsilon} K)$ & $K$  & \\
 & Assignment & $\mathcal{O}(\log^{1-\varepsilon} K)$ & $K$  & \\
 & Selection & $\mathcal{O}(1)$ & $\mathcal{O}(\log K / \log \log K)$ & \Checkmark\\
 & Knapsack & $\mathcal{O}(1)$ & - & \Checkmark \\
\hline
\parbox[t]{2mm}{\multirow{5}{*}{\rotatebox[origin=c]{90}{Regret}}} & Shortest Path &  $\mathcal{O}(\log^{1-\varepsilon} K)$ & $K$  & \Checkmark\\
 & Spanning Tree &  $\mathcal{O}(\log^{1-\varepsilon} K)$ & $K$  & \Checkmark\\
 & $s$-$t$ Cut &  $\mathcal{O}(\log^{1-\varepsilon} K)$ & $K$ &  \\
 & Assignment &  $\mathcal{O}(\log^{1-\varepsilon} K)$ & $K$  & \\
 & Selection & $\mathcal{O}(1)$ & $K$ & \Checkmark \\
 & Knapsack & not approx. & not approx. & 
\end{tabular}
\caption{Current best known approximation guarantees (UB) for unbounded number of scenarios $K$, and best known inapproximability results (LB) (see \cite{kasperski2016robust}).}\label{results}
\end{table}

In this paper a simple improvement of the midpoint approach is presented, where the basic idea is not to aggregate all scenarios into a single scenario, but into a sufficiently small set of scenarios instead. We show that if the min-max problem for a constant number of scenarios is sufficiently approximable, then there is a polynomial-time $\varepsilon K$-approximation for any constant $\varepsilon > 0$. With a slight modification, this also holds for min-max regret. This result hence improves all entries of Table~\ref{results} where the best-known approximation is $K$ and the column ''Fix $K$'' is checked. Interestingly, this also leads to the first-ever approximation algorithm for min-max knapsack problems with unbounded~$K$.

Note that this method is not a PTAS. While PTAS exist for most problems when $K$ is fixed, they have exponential runtime in $K$. Our approach remains polynomial in $K$, but does not give a constant approximation guarantee.

The remainder of this paper is structured as follows. In Section~\ref{minmax} we present our improved approximation algorithm in the case of min-max robustness, and discuss its application to min-max regret in Section~\ref{minmaxregret}. We describe a small computational experiment on our approach in Section~\ref{experiment} before we conclude the paper in Section~\ref{conc}. 

\section{Min-Max Approximation}
\label{minmax}

In this section, we show how to improve the $K$-approximation algorithm for the min-max problem to a $\varepsilon K$-approximation algorithm for any constant $\varepsilon>0$, if a 2-approximation is available for a fixed number of scenarios. The basic idea is the following. Let us assume we have $K=16$ scenarios. Solving the robust problem with all 16 scenarios would yield a 1-approximation (i.e., an optimal solution). Solving the problem with only one aggregated scenario gives a 16-approximation. We show that intermediate scenario aggregations also yield intermediate approximation guarantees (see Figure~\ref{figex}). 

\begin{figure}[htbp]
\centering
\includegraphics[width=\textwidth]{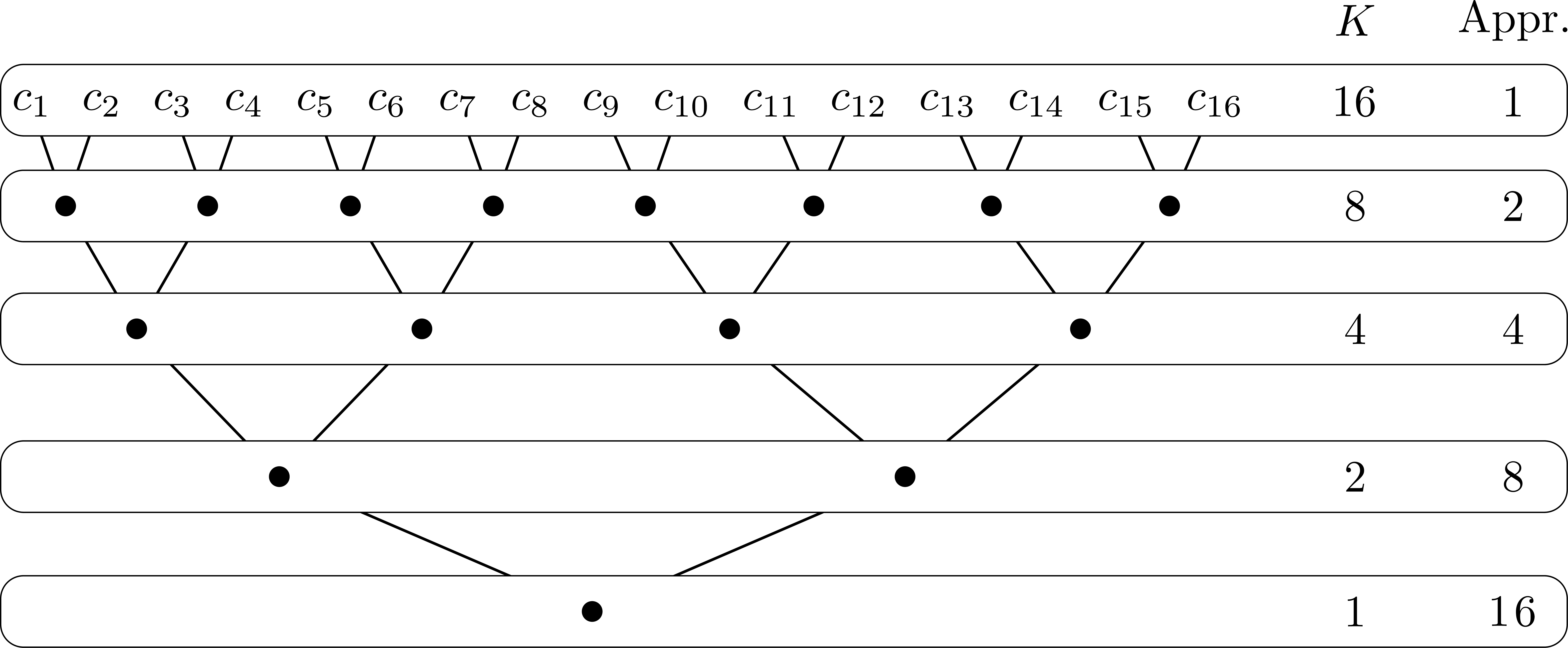}
\caption{Basic aggregation scheme.}\label{figex}
\end{figure}

Now let us assume we would like to have a $K/2$-approximation algorithm. We could aggregate to two scenarios, and solve the resulting problem. However, solving a min-max problem with only two scenarios is usually already NP-hard. Hence, we aggregate to four scenarios instead (which would give a $K/4$-approximation, which is more than we need), and solve this problem with an algorithm that guarantees a 2-approximation. In total, this method then yields a $2\cdot K/4 = K/2$-approximation. In the following, we explain the details of this procedure.

For simplicity, we assume $K = 2^k$ here, but our results readily extend to any $K$. Let any partition of $[K]$ into sets $S_j$ with cardinality $2$ be given $j \in [K/2]$. For each $S_j = \{j_1, j_2\}$, set $\overline{c}_j = \frac{1}{2}( c_{j_1} + c_{j_2} )$, i.e., $\overline{c}_j$ is the midpoint scenario of scenario set $S_j$.

\begin{lemma}
Let $\overline{x}$ be an optimal solution for the min-max problem with scenario set $\overline{\cU} = \{ \overline{c}_1, \ldots, \overline{c}_{K/2} \}$. Then, $\overline{x}$ is a 2-approximation for the min-max problem with scenario set $\cU$.
\label{lem_min_max_half}
\end{lemma}
\begin{proof}
Let $\overline{x}$ be an optimal solution for $\overline{\cU}$, and $x^*$ an optimal solution for $\cU$. Let $i^* = \operatorname{argmax}_{i\in[K]}c_i^t\overline{x}$ be the index of the worst-case scenario in $\cU$ with respect to $\overline{x}$, and choose $j^*$ such that $i^* \in S_{j^*} = \{j^*_1,j^*_2\}$. Then
\begin{align*}
\max_{i\in[K]} c_i^t \overline{x} &= c_{i^*}^t \overline{x} \leq c_{j^*_1}^t \overline{x} +c_{j^*_2}^t \overline{x} 
\le 2 \max_{j\in [K/2]} \left(\frac{c_{j_1} + c_{j_2}}{2}\right)^t \overline{x} \\
&= 2 \max_{j\in[K/2]} \overline{c}_j^t \overline{x} \leq 2 \max_{j\in[K/2]} \overline{c}_j^t x^* \leq 2 \max_{i\in[K]} c_i^t x^* 
\end{align*}
\end{proof}

We repeatedly apply Lemma~\ref{lem_min_max_half} to reduce the number of scenarios. Denote by $\overline{\cU}(k)$ the original scenario set $\cU$ containing all scenarios. After the first level of aggregation we end up with scenario set $\overline{\cU}(k-1)$ containing $2^{k-1}$ scenarios. Repeating the aggregation process we create sets $\overline{\cU}(\ell)$ for $\ell$ from $k$ to $0$.

\begin{corollary}\label{cor}
Applying Lemma~\ref{lem_min_max_half} repeatedly, we get a scenario set $\overline{\cU}(\ell)$ with $2^\ell$ scenarios such that solving the min-max problem with respect to $\overline{\cU}(\ell)$ gives a $2^{k-\ell} = K / 2^\ell$-approximation for the min-max problem with scenario set $\cU$.
\end{corollary}

We present an instance where the approximation guarantee obtained in Corollary~\ref{cor} is tight for the min-max shortest path problem. Let $K=2^k$ be the number of scenarios and $2^\ell$ the number of scenarios that are used in the aggregation. Consider the instance of the shortest path problem presented in Figure~\ref{fig_tight_example}. The top path is divided into $2^\ell$ blocks of $r:=2^{k-\ell}$ edges. All edges in the $i^{\text{th}}$ block have cost $1$ in scenario $(i-1)\cdot r +1$ and $0$ cost in all other scenarios. Hence, the objective value of the top path is equal to $r$. the cost structure for the bottom path is different: The $i^{\text{th}}$ edge of the bottom path has cost $1$ in the $i^{\text{th}}$ scenario and cost $0$ in all other scenarios. Hence, the objective value of the bottom path is $1$. Consider the aggregation schema as in Figure~\ref{figex}. For both paths it holds that after the aggregation the cost of an edge of the $i^{\text{th}}$ block has cost $\frac{1}{r}$ in the $i^{\text{th}}$ aggregated scenario and $0$ in all other aggregated scenarios. Hence, both paths are identical with respect to the aggregated scenarios and the optimal solution of the aggregated problem may consist of the top instead of the bottom path. This leads to a gap of $2^{k-\ell} = K/2^\ell$.

\begin{figure}
\centering
\begin{tikzpicture}[scale= 1.0,
     every state/.style={ 
       minimum size=1em, 
       fill=white, 
       text=black 
     }, 
     node distance=3cm 
   ] 
 	\tikzstyle{empty} = [rectangle,fill = white, minimum size = 0pt,inner sep=0pt];

	\node[state] (s) at (-0.5,0) {$s$}; 

	\node[state] (1t) at (1,1) {}; 
	\node[state] (2t) at (2,1) {}; 
	\node[state] (3t) at (3,1) {}; 
	\node[state] (4t) at (4,1) {}; 
	\node[state] (5t) at (5,1) {}; 
	\node[state] (6t) at (6,1) {}; 
	\node[state] (7t) at (7,1) {}; 
	\node[state] (8t) at (8,1) {}; 
	\node[state] (10t) at (10,1) {}; 

	\node[state] (1b) at (1,-1) {}; 
	\node[state] (2b) at (2,-1) {}; 
	\node[state] (3b) at (3,-1) {}; 
	\node[state] (4b) at (4,-1) {}; 
	\node[state] (5b) at (5,-1) {}; 
	\node[state] (6b) at (6,-1) {}; 
	\node[state] (7b) at (7,-1) {}; 
	\node[state] (8b) at (8,-1) {}; 
	\node[state] (10b) at (10,-1) {}; 
	
	\node[state] (t) at (11.5,0) {$t$}; 
	
	\path[->] (s) edge node[above left] {$e_1$} (1t);     
	\path[->] (s) edge node[below left] {$e_1$} (1b);     
		
	\path[->] (1t) edge node[above] {$e_1$} (2t);     
	\path[->] (1b) edge node[below] {$e_2$} (2b);     

	\path[-,dashed] (2t) edge (3t);     
	\path[-,dashed] (2b) edge (3b);     

	\path[->] (3t) edge node[above] {$e_1$} (4t);     
	\path[->] (3b) edge node[below] {$e_{r}$} (4b);     

	\path[->] (4t) edge node[above] {$e_{r+1}$} (5t);     
	\path[->] (4b) edge node[below] {$e_{r+1}$} (5b);     

	\path[->] (5t) edge node[above] {$e_{r+1}$} (6t);     
	\path[->] (5b) edge node[below] {$e_{r+2}$} (6b);   

	\path[-,dashed] (6t) edge (7t);     
	\path[-,dashed] (6b) edge (7b);       

	\path[->] (7t) edge node[above] {$e_{r+1}$} (8t);     
	\path[->] (7b) edge node[below] {$e_{2r}$} (8b);     

	\path[-,dashed] (8t) edge (10t);     
	\path[-,dashed] (8b) edge (10b);       
   	   	
	\path[->] (10t) edge node[above right] {$e_{K-r+1}$} (t);     
	\path[->] (10b) edge node[below right] {$e_K$} (t);

\end{tikzpicture}
\caption{An instance of the min-max shortest path problem for which the approximation guarantee of Corollary~\ref{cor} is tight. All edges share a similar cost structure: The cost is $1$ in one scenario and $0$ in all other scenarios. We represent this by the unit vectors $e_1,\dots,e_K$.}
\label{fig_tight_example}
\end{figure}
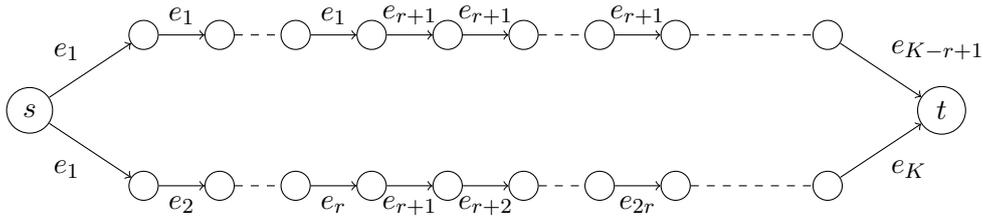

\begin{lemma}\label{lemma2}
A solution that is an $\alpha$-approximation for $\overline{\cU}(\ell)$ is also an $(\alpha K /2^\ell$)-approximation for~$\cU$.
\end{lemma}
\begin{proof}
Analogously to the proof of Lemma~\ref{lem_min_max_half}.
\end{proof}

\noindent
We can now state the main result of this section.

\begin{theorem}
Let a constant $0 \le \varepsilon \leq 1$ be given. If there exists a 2-approximation algorithm for the min-max problem with a fixed number of scenarios, there exists a polynomial-time algorithm that gives an $(\varepsilon K)$-approximation for the min-max problem.
\label{thm_eps_k_approx_min_max}
\end{theorem}
\begin{proof}
Let $\varepsilon \leq 1$ be constant. We choose $\ell' := \lceil\log \frac{1}{\varepsilon} + 1 \rceil$. According to Corollary~\ref{cor} we construct the set $\overline{\cU}(\ell')$ with $2^{\ell'}$ scenarios. Using the 2-approximation algorithm for the min-max problem with a fixed number of scenarios, we find a 2-approximation for $\overline{\cU}(\ell')$. Using Lemma~\ref{lem_min_max_half}, we conclude that the solution is a $2\cdot 2^{k-\ell} = 2K/ 2^{\lceil \log\frac{1}{\varepsilon}+1 \rceil} \leq  K/\frac{1}{\varepsilon} = \varepsilon K$-approximation. Note that the running time of this procedure is polynomial since the value of $\epsilon$ and, therefore, also $\ell'$, is fixed.
\end{proof}

\begin{corollary}
For the min-max shortest path, spanning tree, selection, and knapsack problem with unbounded number of scenarios $K$, there exists a polynomial-time $(\varepsilon K)$-approximation algorithm for any fixed $\varepsilon > 0$.
\end{corollary}

As examples, let us consider the min-max selection and shortest path problems. For both problems we need a 2-approximation algorithm for the min-max problem with a fixed number of scenarios.
\begin{itemize}
\item For selection, there exists an FPTAS that finds for a fixed number of scenarios $\tilde{K}$ a $(1+\tilde{\varepsilon})$-approximation with running time $\mathcal{O}\left(\frac{p^{\tilde{K}}n}{\tilde{\varepsilon}^{\tilde{K}-1}}\right)$ (see \cite{kasperski2007approximation}). Hence, using Theorem~\ref{thm_eps_k_approx_min_max}, an $(\varepsilon K)$-approximation is possible in $\mathcal{O}\left(p^{\frac{4}{\varepsilon}}n\right)$ by aggregating to $\tilde{K}= 2^{\lceil\log \frac{1}{\varepsilon} + 1 \rceil}$ scenarios, and approximating the resulting problem with a factor $2$, i.e., choosing $\tilde{\varepsilon} = 1$.

\item For shortest path, there exists an FPTAS that finds for a fixed number of scenarios $\tilde{K}$ a $(1+\tilde{\varepsilon})$-approximation with running time in $\mathcal{O}\left(\frac{nm^{\tilde{K}}}{\tilde{\varepsilon}^{\tilde{K}-1}}\right)$, which makes it possible to find an $(\varepsilon K)$-approximation in time $\mathcal{O}\left(mn^{\frac{4}{\varepsilon}}\right)$ by aggregating to $\tilde{K}= 2^{\lceil\log \frac{1}{\varepsilon} + 1 \rceil}$ scenarios, and approximating the resulting problem with a factor $2$, i.e. choosing~$\tilde{\varepsilon} = 1$.
\end{itemize}

\section{Min-Max Regret Approximation}
\label{minmaxregret}

To translate the results obtained for min-max to min-max regret problems we need to modify the aggregation procedure, as the following example shows.

\begin{example}\label{ex1}
Consider the min-max regret shortest path instance shown in Figure~\ref{ex1-1}. There are four scenarios. An optimal solution is to take the path in the middle with a regret of $1$. If we aggregate the first two and the last two scenarios, we arrive at the instance shown in Figure~\ref{ex1-2}. Here, an optimal solution (with perceived regret $1$) is to take the top path; but its true regret is $4$. Hence, in this example, we obtain only a $4$-approximation and not a $2$-approximation as in the case of the min-max objective function.
\end{example}

\begin{figure}[htbp]
\centering
\begin{subfigure}[b]{0.45\textwidth}
\centering
\begin{tikzpicture}[scale= 1.5,
     every state/.style={ 
       minimum size=2em, 
       fill=white, 
       text=black 
     }, 
     node distance=3cm 
   ] 
   
	\tikzstyle{empty} = [rectangle,fill = white, minimum size = 0pt,inner sep=0pt];

	\node[state] (s) at (0,0) {$s$}; 
    \node[state] (t) at (3,0) {$t$}; 
	
	\node[empty] (e1) at (1.5,1.0) {$(4,0,0,0)$};    
	\node[empty] (e2) at (1.5,-0.5) {$(0,4,0,0)$};    
   
	\path[->] (s) edge node[above] {$(1,1,1,1)$} (t);     
	\path[->, bend left = 50] (s) edge (t);
	\path[->, bend right = 50] (s) edge (t);
    	
\end{tikzpicture}
\subcaption{Instance with 4 scenarios.}
\label{ex1-1}
\end{subfigure}
\begin{subfigure}[b]{0.45\textwidth}
\centering
\begin{tikzpicture}[ scale= 1.5,
     every state/.style={ 
       minimum size=2em, 
       fill=white, 
       text=black 
     }, 
     node distance=3cm 
   ] 
   
	\tikzstyle{empty} = [rectangle,fill = white, minimum size = 0pt,inner sep=0pt];

	\node[state] (s) at (0,0) {$s$}; 
    \node[state] (t) at (3,0) {$t$}; 
	
	\node[empty] (e1) at (1.5,1.0) {$(2,0)$};    
	\node[empty] (e2) at (1.5,-0.5) {$(2,0)$};    
        
	\path[->] (s) edge node[above] {$(1,1)$} (t);     
	\path[->, bend left = 50] (s) edge (t);
	\path[->, bend right = 50] (s) edge (t);
    	
\end{tikzpicture}
\subcaption{Instance with aggregated scenarios}
\label{ex1-2}
\end{subfigure}

\caption{An example instance where aggregating from four to two scenarios leads to a $4$-approximation for min-max regret.}
\end{figure}
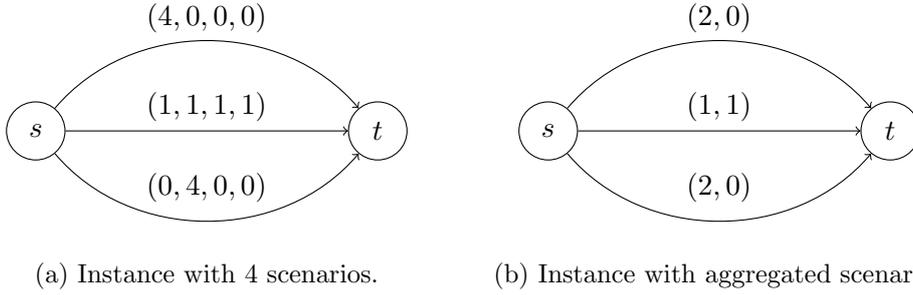

Instead of simply aggregating pairs $S_j = \{j_1,j_2\}$ of scenarios and solving a min-max regret problem on this new scenario set, we consider the following problem 
\begin{align*}
\min_{x\in\X} \max_{j\in[K/2]} \overline{c}_j^t x - \frac{1}{2} opt(c_{j_1}) - \frac{1}{2} opt(c_{j_2}) \tag{*} 
\end{align*}
Note that we do not use the objective $\overline{c}^t_j - opt(\overline{c}_j)$, as would be usual for min-max regret. Further generalizing, we call a problem
\[ \min_{x\in\X} \max_{i\in[K]} c_i^t x - d(c_i) \]
with arbitrary $d$ a generalized min-max regret problem.

\begin{lemma}\label{lem_min_max_regret_half}
Solving the generalized min-max regret problem (*) on $\overline{\cU} = \{ \overline{c}_1, \ldots, \overline{c}_{K/2} \}$ is a 2-approximation for $\cU$.
\end{lemma}
\begin{proof}
Let $\overline{x}$ be optimal for problem (*), and $x^*$ optimal for the original problem with uncertainty set $\cU$. Again, denote by  $i^* = \operatorname{argmax}_{i\in[K]}c_i^t\overline{x} -opt(c_i)$ and choose $j^*$ such that $i^* \in S_{j^*}= \{j^*_1,j^*_2\}$. Then 
\begin{align*}
\max_{i\in[K]} c_i^t \overline{x} - opt(c_i) &= c_{i^*}^t \overline{x} - opt(c_{i^*})  \\
&\leq c_{j^*_1}^t \overline{x} - opt(c_{j^*_1}) + c_{j^*_2}^t \overline{x} - opt(c_{j^*_2}) \\
&\leq \max_{j\in[K/2]}  ( c_{j_1} + c_{j_2} )^t \overline{x} - opt(c_{j_1}) - opt(c_{j_2}) \\
&= 2 \max_{j\in[K/2]}  \overline{c}_j^t \overline{x} - \frac{1}{2}opt(c_{j_1}) - \frac{1}{2}opt(c_{j_2})  \\
&\leq  2 \max_{j\in[K/2]}  \overline{c}_j^t x^* - \frac{1}{2}opt(c_{j_1}) - \frac{1}{2}opt(c_{j_2})  \\
&= \max_{j\in[K/2]}c_{j_1}^t x^* - opt(c_{j_1})  + c_{j_2}^t x^* -  opt(c_{j_2}) \\
& \leq 2 \max_{i\in[K]}  c_i^t x^* - opt(c_i) 
\end{align*}
\end{proof}

Note that the arguments used in the proof of Lemma~\ref{lem_min_max_regret_half} can be generalized to the case where  scenarios are aggregated repeatedly, i.e., we aggregate to sets $S_j$ with more than two elements. Similar to Corollary~\ref{cor} we obtain:

\begin{corollary}
Given an aggregated scenario set $\overline{\cU}(\ell)$ where each of the $2^\ell$ scenarios is given as $\overline{c}_j := \frac{1}{2^{k-\ell}}\sum_{s \in S_j} c_s$.  The optimal solution of the generalized min-max regret problem
\begin{align*}
\min_{x\in\X} \max_{j\in[2^\ell]} \overline{c}_j x - \frac{1}{2^{k-\ell}}  \sum_{s \in S_j} opt(c_k)
\end{align*}
yields an $(K/2^\ell)$-approximation of the min-max regret problem with scenario set $\mathcal{U}$.
\label{cor2}
\end{corollary}

Similar to min-max, we can use a 2-approximation for the generalized min-max regret problem with a fixed number of scenarios to obtain an $(\epsilon K)$-approximation for the min-max regret problem. Using the same arguments as in the proof of Theorem~\ref{thm_eps_k_approx_min_max} we obtain:

\begin{theorem}
Let a constant $0 < \varepsilon \leq 1$ be given. If there exists a 2-approximation algorithm for the generalized min-max regret problem with a fixed number of scenarios, then there exists a polynomial-time algorithm that gives an $(\varepsilon K)$-approximation for the min-max regret problem.
\label{thm_eps_k_approx_min_max_regret}
\end{theorem}

Note that in our construction of the generalized problem, we have $d(c_i) \le opt(c_i)$. Hence, we can use the same proof as in \cite{Aissi2007281} using the FPTAS for multi-objective spanning tree to show that there is an FPTAS for our generalized min-max regret spanning tree problem. The same approach applies to the min-max regret selection problem.

Furthermore, we can modify any generalized min-max regret shortest path problem by adding an edge from $s$ to $t$ with costs $d(c_i)$ in scenario $i$. We create an additional scenario where the costs of each edge is 0, and the costs of the new edge is a sufficiently large value $M$. As $d(c_i) \le opt(c_i)$, we can then solve a classic min-max regret problem on this instance, giving the same objective value as before. Hence, the FPTAS for min-max regret shortest path (see \cite{Aissi2007281}) can also be applied to our generalized problem.

\begin{corollary}
For the min-max regret shortest path, spanning tree, and selection problem with unbounded number of scenarios $K$, there exists a polynomial-time $(\varepsilon K)$-approximation algorithm for any fixed $\varepsilon > 0$.
\end{corollary}

\section{Computational Experiments}
\label{experiment}

In this section, we present a small test of the proposed aggregation method on randomly generated instances. As benchmark problem we use the shortest path problem. We define a complete layered graph $G$ with $10$ layers and width $4$. The scenario set consists of $16$ randomly generated scenarios. The cost of each edge is chosen uniformly in $[0,1]$ for each scenario. In the first step, we begin with the full set of scenarios. Next, we half the number of scenarios by aggregating them pairwise. This is repeated till we end up with a single scenario. In each step, we solve the min-max shortest path problem with the corresponding set of scenarios. To solve these problems, we solve an IP formulation of the problem using CPLEX. Note that in the last step, where the uncertainty set consists only of a single scenario, only a classic shortest path problem needs to be solved. At the end, we evaluate the performance of the computed paths by computing their worst case cost using the original set of $16$ scenarios. To average the results we divide the performance of each solution by the performance of the optimal solution. 

The aggregation scheme presented in Figure~\ref{figex} proposes to aggregate always two consecutive scenarios. This is arbitrary, and any other aggregation rule can be used in practice to improve the performance of the method. Beside the aggregation of consecutive scenarios, we also tested to aggregate similar scenarios. To this end, we computed a perfect matching between the different scenarios. We set the cost of matching scenario $i$ with scenario $j$ to the euclidean distance of scenario $i$ and $j$. 

The results of the experiment, averaged over 1000 instances, are shown in Figure~\ref{fig_approx_exp}.

\begin{figure}
\centering
\includegraphics[width=0.8\textwidth,clip, trim = 60 50 90 70]{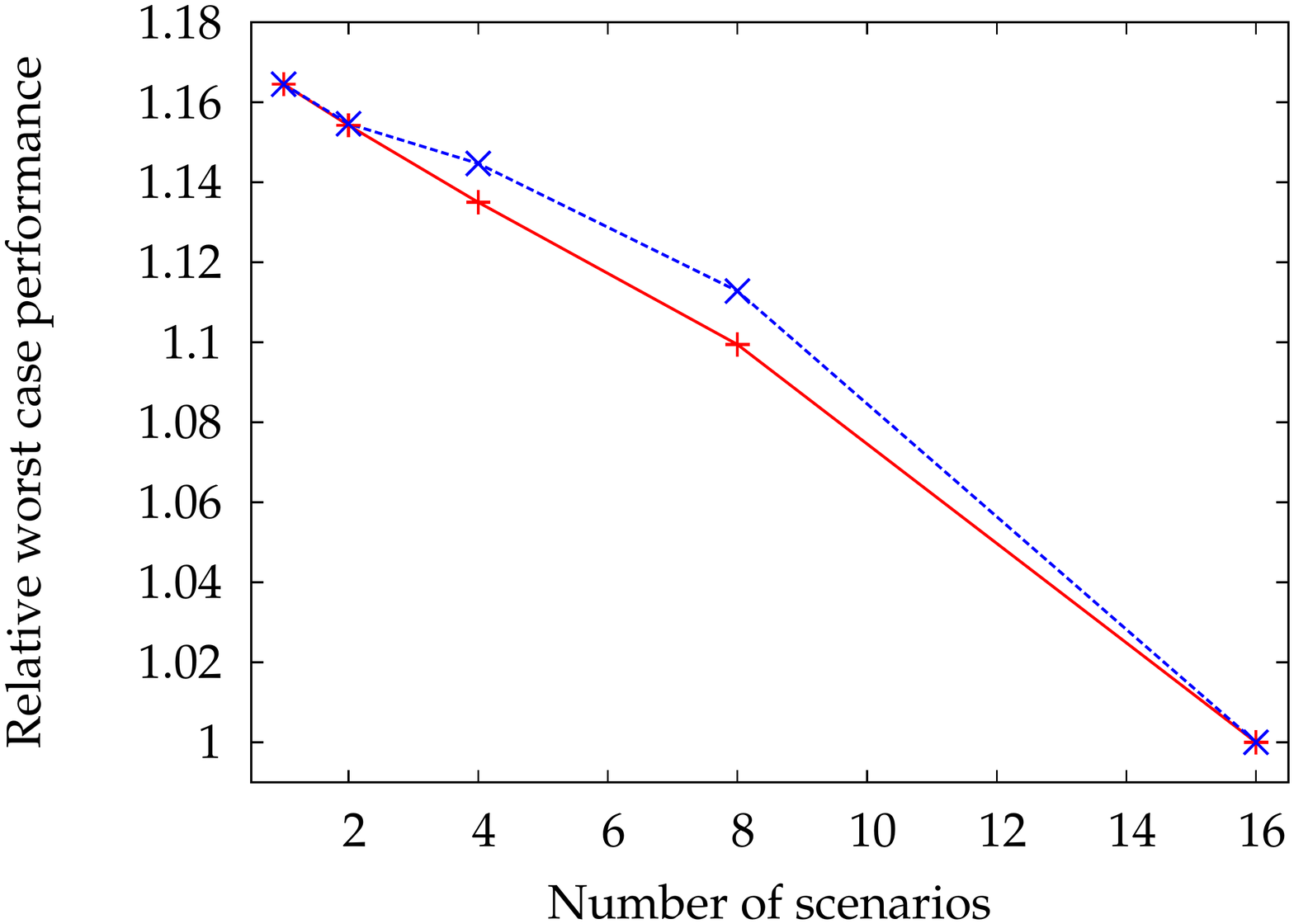}
\caption{The horizontal axis gives the number of scenarios that are used for the different aggregation levels. The relative worst case performance of the different solutions for the different levels of aggregation is shown on the vertical axis. The red straight line shows the performance when aggregating similar scenarios, and the blue dashed line shows the performance of the consecutive aggregation scheme.}
\label{fig_approx_exp}
\end{figure}

It can be seen that the more involved aggregation rule based on scenario similarity does indeed gives better results for intermediate aggregation levels. For full or no aggregation, the aggregation rule is of course irrelevant. Note also that the relative performance of the aggregated solutions is far below the theoretical performance guarantee. Interestingly, for the aggregation scheme based on similarity, there seems to exist a roughly linear relationship between the number of scenarios and the relative worst case performance.

\section{Conclusions}
\label{conc}

The midpoint method is a central approximation algorithm in robust optimization. Despite its simplicity, is has been the best-known method for several classic combinatorial problems. In this paper we presented a simple variant of the method, where the uncertainty set is not aggregated to a single scenario, but to a sufficiently small set of scenarios instead. This reduced scenario set is then approximated using, e.g., an FPTAS for discrete uncertainty of constant size. Our approach can be used to find polynomial time $(\varepsilon K)$-approximations for any constant $\varepsilon\in[0,1]$, thus improving several currently known best approximability results.

Our results hold for any aggregation scheme. However, for practical purposes, aggregating similar scenarios is reasonable, so as to preserve the structure of the uncertainty set as far as possible. To quantify this effect we considered a computational experiment using random shortest path instances. Our results indicate that approximation guarantees are considerably smaller on these instances than the theoretical bounds suggest, and that aggregating similar scenarios does indeed improve the quality of solutions.

\end{document}